\documentclass[a4paper,10pt]{article}
\usepackage{amsthm,amsmath,amssymb}
\usepackage{graphicx}
\usepackage{graph}
\usepackage{tikz}
\newtheorem{Theorem}{Theorem}
\newtheorem{Definition}{Definition}
\newtheorem{Lemma}{Lemma}
\newtheorem{Corollary}{Corollary}
\newtheorem{Remark}{Remark}
\newtheorem{Observation}{Observation}

\newcommand{\com}[1]{}
\newcommand{\boxi}{\mathrm{box}}
\newcommand{\tw}{\mathrm{tw}}
\newcommand{\pw}{\mathrm{pw}}

\begin{document}
\title{Representing graphs as the intersection of cographs\\and threshold graphs}
\bibliographystyle{siam}
\author{Daphna Chacko\thanks{Dept. of Computer Science and Engineering, National Institute of Technology, Calicut. e-mail: \texttt{daphna.chacko@gmail.com}} \and Mathew C. Francis\thanks{Indian Statistical Institute, Chennai Centre. e-mail: \texttt{mathew@isichennai.res.in}}}
\date{}
\maketitle
\begin{abstract}
A graph $G$ is said to be the intersection of graphs $G_1,G_2,\ldots,G_k$ if $V(G)=V(G_1)=V(G_2)=\cdots=V(G_k)$ and $E(G)=E(G_1)\cap E(G_2)\cap\cdots\cap E(G_k)$.
For a graph $G$, $\dim_{COG}(G)$ (resp. $\dim_{TH}(G)$) denotes the minimum number of cographs (resp. threshold graphs) whose intersection gives $G$. We present several new bounds on these parameters for general graphs as well as some special classes of graphs.
It is shown that for any graph $G$: (a) $\dim_{COG}(G)\leq\tw(G)+2$, (b) $\dim_{TH}(G)\leq\pw(G)+1$, and (c) $\dim_{TH}(G)\leq\chi(G)\cdot\boxi(G)$, where $\tw(G)$, $\pw(G)$, $\chi(G)$ and $\boxi(G)$ denote respectively the treewidth, pathwidth, chromatic number and boxicity of the graph $G$. We also derive the exact values for these parameters for cycles and show that every forest is the intersection of two cographs. These results allow us to derive improved bounds on $\dim_{COG}(G)$ and $\dim_{TH}(G)$ when $G$ belongs to some special graph classes.
\end{abstract}
\section{Introduction}

All graphs in this paper are simple, finite and undirected, unless otherwise mentioned. Let $G(V,E)$ be a graph, where $V(G)$ is the vertex set and $E(G)$ is the edge set of $G$. Given graphs $G_1,G_2,\ldots,G_k$ such that $V(G)=V(G_1)=V(G_2)=\cdots=V(G_k)$, we say that $G=G_1\cap G_2\cap\cdots\cap G_k$ if $E(G)=E(G_1)\cap E(G_2)\cap\cdots\cap E(G_k)$; in this case we say that ``$G$ is the intersection of the graphs $G_1,G_2,\ldots,G_k$'', or that ``$G$ can be represented as the intersection of $k$ graphs $G_1,G_2,\ldots,G_k$''.

Let $\mathcal{A}$ be a class of graphs. We are concerned with the question of representing a graph $G$ as the intersection of a small number of graphs from $\mathcal{A}$. Kratochv\'il and Tuza~\cite{kratochvil1994intersection} defined the \emph{intersection dimension} of a graph $G$ with respect to a graph class $\mathcal{A}$, denoted by $\dim_{\mathcal{A}}(G)$, as the smallest number of graphs from $\mathcal{A}$ whose intersection gives $G$. This is formally defined below.

\begin{Definition}[Intersection dimension of graph~\cite{kratochvil1994intersection}]
Given a class $\mathcal{A}$ of graphs and a graph $G(V,E)$, the intersection dimension of $G$ with respect to $\mathcal{A}$ 
is defined as:
\begin{center}
$\dim_{\mathcal{A}}(G) = \min \{k\colon \exists G_1,G_2,\ldots,G_k\in\mathcal{A}$ such that 
$G= \bigcap\limits_{i=1}^{k} G_i\}$
\end{center}
\end{Definition}

Kratochv\'il and Tuza also note that for a graph class $\mathcal{A}$, $\dim_{\mathcal{A}}(G)$ exists for every graph $G$ if and only if $\mathcal{A}$ contains all complete graphs and all graphs that can be obtained by removing an edge from a complete graph.
The notion of intersection dimension was introduced as a generalization of some well-studied notions like boxicity, circular dimension and overlap dimension of graphs (see~\cite{kratochvil1994intersection}).

A ``complement reducible graph'' or \emph{cograph} is a graph that can be recursively constructed from copies of $K_1$ (the graph containing one vertex and no edges) using the disjoint union and complementation operations. Cographs turn out to be exactly those graphs that do not contain a $P_4$---a path on four vertices---as an induced subgraph~\cite{corneil1981complement}.

A \emph{split graph} $G$ is a graph whose vertices can be partitioned into two sets, one of which is an independent set in $G$ and the other a clique in $G$. Cographs that are also split graphs form the class of graphs known as \emph{threshold graphs}~\cite{Brandstadt:1999:GCS:302970}. Threshold graphs have been widely studied in the literature and have several different equivalent definitions (see~\cite{mahadev1995threshold,golumbic}). For example, they are exactly the graphs that do not contain an induced subgraph isomorphic to a $P_4$, $2K_2$ (the graph having four vertices and two disjoint edges) or $C_4$ (the cycle on four vertices)~\cite{chvatalhammer}. In fact, this characterization follows from the fact that split graphs are exactly those graphs that do not contain a $2K_2$, $C_4$ or $C_5$ (the cycle on five vertices) as an induced subgraph~\cite{foldes1977split}.

In this paper, we use $COG$ and $TH$ to denote the class of cographs and the class of threshold graphs respectively.
Thus, $\dim_{COG}(G)$ is the intersection dimension of a graph $G$ with respect to cographs, and for short, we call this the ``cograph dimension of $G$''. Similarly, we shall call $\dim_{TH}(G)$ the ``threshold dimension of $G$''.

We say that a graph $G$ ``is covered by'' graphs $G_1,G_2,\ldots,G_k$, if $V(G)=V(G_1)=V(G_2)=\cdots=V(G_k)$ and $E(G)=E(G_1)\cup E(G_2)\cup\cdots\cup E(G_k)$. Note that since cographs and threshold graphs are both closed under complementation, it follows that $\dim_{COG}(\overline{G})$ and $\dim_{TH}(\overline{G})$ denote respectively the smallest number of cographs, respectively threshold graphs, using which a graph $G$ can be covered. Chv\'atal and Hammer introduced the parameter $t(G)$, defined as the smallest number of threshold graphs required to cover a graph $G$~\cite{chvatalhammer}, the study of which has resulted in several influential papers~\cite{yannakakis1982complexity,cozzens1984threshold}. The parameter $t(G)$ has been called the ``threshold dimension'' of $G$ due to the equivalent definition of this parameter as the smallest number of linear inequalities on $|V(G)|$ variables such that every inequality is satisfied by a vector in $\{0,1\}^{|V(G)|}$ if and only if it is the characteristic vector of an independent set in $G$ (refer~\cite{raschlesimon} for details). The parameter $t(G)$ is also known as the \emph{threshold cover} of $G$~\cite{cozzens1984threshold}. In this paper, we shall refer to $t(G)$ exclusively as the ``threshold cover'' of $G$. We reserve the term ``threshold dimension of $G$'' for $\dim_{TH}(G)$. Since  $t(G)=\dim_{TH}(\overline{G})$, our results about $\dim_{TH}(G)$ can also be thought of as results about $t(\overline{G})$. Thus $\dim_{TH}(G)$ is the minimum integer $k$ such that there exist $k$ linear inequalities
\begin{eqnarray*}
a_{11}x_1+a_{12}x_2+\cdots+a_{1n}x_n&\leq&t_1\\
a_{21}x_1+a_{22}x_2+\cdots+a_{2n}x_n&\leq&t_2\\
\vdots\\
a_{k1}x_1+a_{k2}x_2+\cdots+a_{kn}x_n&\leq&t_k
\end{eqnarray*}
on the variables $x_1,x_2,\ldots,x_n$, where $n=|V(G)|$, such that the characteristic vector of a set $S\subseteq V(G)$ satisfies all the inequalities if and only if $S$ is a clique in $G$. In other words, it is the minimum number of halfspaces in $\mathbb{R}^n$ whose intersection contains exactly those corners of the $n$-dimensional hypercube that correspond to cliques in $G$ (the corners of the $n$-dimensional hypercube are the points in $\{0,1\}^n$).

Hellmuth and Wieseke~\cite{hellmuth2015symbolic} showed that the problem of determining whether the edge set of an input graph can be written as the union of the edge sets of 2 cographs is NP-complete. As cographs are closed under complementation, this implies that the problem of determining whether $\dim_{COG}(G)\leq 2$ for an input graph $G$ is NP-complete. Yannakakis~\cite{yannakakis1982complexity} showed that the problem of determining if an input graph $G$ has $t(G)\leq k$ is NP-complete for every fixed $k\geq 3$. As $\dim_{TH}(G)=t(\overline G)$, this means that it is NP-hard to determine if the threshold dimension of a graph is at most $k$ for every fixed $k\geq 3$. Raschle and Simon~\cite{raschlesimon} showed that it can be decided in polynomial-time whether $t(G)\leq 2$ for an input graph $G$; thus the problem of checking whether an input graph has threshold dimension at most 2 is solvable in polynomial time.
Gimbel and Ne\v{s}et\v{r}il~\cite{gimnes} discuss the problem of deciding whether the vertex set of an input graph $G$ can be partitioned into $k$ parts such that the subgraph induced in $G$ by each part is a cograph and show that this problem is NP-complete for every fixed $k\geq 2$.

In this paper, we show three upper bounds for the cograph dimension and threshold dimension of general graphs in terms of their treewidth, pathwidth, chromatic number and boxicity. In particular, we show that for any graph $G$,
\begin{itemize}
\item $\dim_{COG}(G)\leq\tw(G)+2$,
\item $\dim_{TH}(G)\leq\pw(G)+2$, and
\item $\dim_{TH}(G)\leq\chi(G)\cdot\boxi(G)$.
\end{itemize}
Here $\tw(G)$, $\pw(G)$, $\chi(G)$ and $\boxi(G)$ denote respectively the treewidth, pathwidth, chromatic number and boxicity of the graph $G$ (treewidth and pathwidth are defined in Section~\ref{sec:cogdimtwd}, whereas chromatic number and boxicity are defined in Section~\ref{sec:threshboxichro}). Note that the upper bound of $\tw(G)+2$ on the cograph dimension of any graph $G$ is equal to the upper bound on boxicity proved in~\cite{chandran2007boxicity}. This is interesting considering that the boxicity of a graph $G$ is nothing but $\dim_{INT}(G)$ where $INT$ is the class of interval graphs.

Kratochv\'il and Tuza~\cite{kratochvil1994intersection} show that $\dim_{PER}(G)\leq 12$ (here, $PER$ denotes the class of permutation graphs) when $G$ is a planar graph and ask whether this bound can be improved. Since cographs and threshold graphs are subclasses of permutation graphs, any upper bound on $\dim_{COG}(G)$ or $\dim_{TH}(G)$ is also an upper bound on $\dim_{PER}(G)$. We show that if $G$ is a planar graph, then $\dim_{TH}(G)\leq 12$ and $\dim_{COG}(G)\leq 10$. We further give better upper bounds for both the parameters when $G$ belongs to some special subclasses of planar graphs. For example, for the case when $G$ is a planar bipartite graph, we improve the bound $\dim_{PER}(G)\leq 4$ given in~\cite{kratochvil1994intersection} to $\dim_{TH}(G)\leq 4$. These results are summarized in Table~\ref{tab:results} in Section~\ref{sec:conclusion}. We also show that cycles on more than 6 vertices cannot be represented as the intersection of two cographs, which allows us to derive the exact values for the cograph dimension and threshold dimension of every cycle.

The paper is organized as follows. We begin by studying the threshold and cograph dimensions of forests in Section~\ref{sec:forests}, and that of cycles in Section~\ref{sec:cycles}. We then derive the upper bound on threshold dimension in terms of boxicity and chromatic number in Section~\ref{sec:threshboxichro}. In Section~\ref{sec:cogdimtwd}, we derive upper bounds on the cograph dimension and threshold dimension of a graph in terms of its treewidth and pathwidth respectively. In Section~\ref{sec:vertpart}, using a modified form of a lemma from~\cite{kratochvil1994intersection}, we explore the connection between the acyclic chromatic number and star chromatic number of a graph and its cograph dimension, which was first investigated in~\cite{aravindcrs}. By combining this with the result obtained in Section~\ref{sec:forests}, we get new upper bounds on the cograph dimension for planar graphs and for planar graphs with lower bounds on girth.
In order to show that the technique using the star chromatic number cannot yield a bound on the cograph dimension of outerplanar graphs that is better than the one obtained using treewidth, we present an outerplanar graph whose star chromatic number can be proved to be at least 6. This graph, in our opinion, is simpler than the earlier known example, for which only a computer-aided proof is available. Some open questions and a table summarising the various upper bounds obtained for the cograph dimension and threshold dimension of the subclasses of planar graphs that are studied is given in Section~\ref{sec:conclusion}.
\medskip

The following section contains some preliminary observations and definitions which will be used later.

\section{Preliminaries}
Given a graph $G(V,E)$, we let $V(G)$ and $E(G)$ denote its vertex set and edge set respectively. For a vertex $u\in V(G)$, $N(u)$
denotes the set of neighbours of $u$ and $N[u]=N(u)\cup\{u\}$. Given a set $S\subseteq V(G)$, we denote by $G[S]$ the subgraph induced in $G$ by the vertices in $S$. An \emph{induced $P_4$ (resp. $2K_2$, $C_4$)} in $G$ is an induced subgraph of $G$ that is isomorphic to a $P_4$ (resp. $2K_2$, $C_4$).
\begin{Observation}
Let $\mathcal{A}$ be a class of graphs and let $G,G_1,G_2,\ldots,G_k$ be graphs such that $G=G_1\cap G_2\cap\cdots\cap G_k$. Then $\dim_{\mathcal{A}}(G)\leq \dim_{\mathcal{A}}(G_1)+\dim_{\mathcal{A}}(G_2)+\cdots+\dim_{\mathcal{A}}(G_k)$.
\end{Observation}
\begin{proof}
For each $i\in\{1,2,\ldots,k\}$, there exists a collection $\mathcal{H}_i\subseteq\mathcal{A}$ of at most $\dim_{\mathcal{A}}(G_i)$ graphs such that $G_i=\bigcap_{H\in\mathcal{H}_i} H$. Since $G=\bigcap_{i=1}^k G_i=\bigcap_{i=1}^k\bigcap_{H\in\mathcal{H}_i} H$, we now have that $\dim_{\mathcal{A}}(G)\leq\sum_{i=1}^k \dim_{\mathcal{A}}(G_i)$.
\end{proof}
\begin{Definition}[the join operation]
The \emph{join} of two vertex-disjoint graphs $G_1$ and $G_2$, denoted as $G_1+G_2$, is the graph $G$ having $V(G)=V(G_1)\cup V(G_2)$ and $E(G)=E(G_1)\cup E(G_2)\cup \{uv\colon u\in E(G_1), v\in E(G_2)\}$.
\end{Definition}

The join operation is called the ``Zykov sum'' operation in~\cite{kratochvil1994intersection}.

\begin{Definition}[the disjoint union operation]
The \emph{disjoint union} of two vertex-disjoint graphs $G_1$ and $G_2$, denoted as $G_1\uplus G_2$, is the graph $G$ having $V(G)=V(G_1)\cup V(G_2)$ and $E(G)=E(G_1)\cup E(G_2)$.
\end{Definition}

A class of graphs $\mathcal{A}$ is said to be ``closed'' under the join operation (resp. the disjoint union operation) if for any $G_1,G_2\in\mathcal{A}$, we have $G_1+G_2\in\mathcal{A}$ (resp. $G_1\uplus G_2\in\mathcal{A}$). The following observation is easy to see.

\begin{Observation}\label{obs:join}
Let $\mathcal{A}$ be a class of graphs that is closed under the join operation (resp. disjoint union operation). Then for graphs $G_1$ and $G_2$, $\dim_{\mathcal{A}}(G_1+G_2)$ (resp. $\dim_{\mathcal{A}}(G_1\uplus G_2))\leq\max\{\dim_{\mathcal{A}}(G_1),\dim_{\mathcal{A}}(G_2)\}$.
\end{Observation}

Note that the class of cographs is closed under both the join and disjoint union operations whereas the class of threshold graphs is not closed under either operation. A class of graphs is said to be \emph{hereditary} if it is closed under taking induced subgraphs; for any graph $G$ in the class, every induced subgraph of $G$ is also in the class. From the definition of cographs, it follows that the only hereditary class of graphs that is closed under both the join and disjoint union operations is the class of cographs.

Partial 2-trees are exactly the graphs that have treewidth at most 2~\cite{Bodlaender98}. Outerplanar graphs are the planar graphs that have a planar embedding in which every vertex is on the boundary of the outer face. They form a subclass of partial 2-trees~\cite{Bodlaender98}. For any terminology or notation that is not defined herein, please refer~\cite{Diestel}.
\section{Forests}\label{sec:forests}
In this section, we shall show that every forest is the intersection of at most two cographs and then, using a known characterization of graphs that have a threshold cover of size 2, we show that there exist trees that are not the intersection of two threshold graphs. 

We first show a construction using which given any forest $F$, two cographs whose intersection gives $F$ can be constructed. We describe the construction for a tree, and the construction for forests as an extension to it.

Let $T$ be a tree in which one vertex $r$ has been arbitrarily selected to be the 
root. The ancestor-descendant and parent-child relations on $V(T)$ are then 
defined in the usual way (i.e., for $x,y\in V(T)$, $x$ is an ancestor of $y$ if and only if $x$ lies on the path between $r$ and $y$ in $T$; $x$ is the parent of $y$ if and only if $x$ is an ancestor of $y$ and $xy\in E(T)$). The vertices that are at an even distance from the 
root $r$ are called ``even vertices'' and those that are at an odd distance from 
the root $r$ are called ``odd vertices''. In the following, for any vertex $v\in 
V(T)$, we denote by $p(v)$ its parent vertex. We let $p(r)=r$.

Let $T_o$ (resp. $T_e$) denote the graph with vertex set $V(T_o)$ (resp. $V(T_e)$)
$=V(T)$ and $E(T_o)$ (resp. $E(T_e)$) $=E(T)\cup\{
uv\colon u$ is an odd (resp. even) vertex and $v$ is a descendant of $p(u)\}$.

\begin{Lemma}\label{lem:compiscograph}
Let $T$ be a tree with a root. Then both $T_o$ and $T_e$ are cographs.
\end{Lemma}
\begin{proof}
We shall first prove that $T_o$ is a cograph using induction on $|V(T)|$.
Clearly, $T_o$ is a cograph when $|V(T)|=1$, since $K_1$ is a cograph.
Let $T^1,T^2,\ldots,T^k$ be the trees which form the components
of $T-N[r]$. It is easy to see that for each $i\in\{1,2,\ldots,k\}$, there exists
exactly one vertex $r_i$ in $T^i$ such that in $T$, $p(p(r_i))=r$. Choose $r_i$ to be
the root of $T^i$, for each $i\in\{1,2,\ldots,k\}$. Observe that $T_o=(T[\{r\}]\uplus
T^1_o\uplus T^2_o\uplus\cdots\uplus T^k_o)+\sum_{u\in N(r)} T[\{u\}]$ (here, for a
vertex $v\in V(T)$, $T[\{v\}]$ refers to the subgraph of $T$ isomorphic to $K_1$
that contains just the vertex $v$). By the induction hypothesis, $T^1_o,T^2_o,\ldots,
T^k_o$ are all cographs. Since $K_1$ is a cograph and cographs are closed under the
join and disjoint union operations, we have that $T_o$ is a cograph.

Next let us prove that $T_e$ is a cograph.
Let $T^1,T^2,\ldots,T^k$ be the trees that form the components of $T-r$. For each
$i\in\{1,2,\ldots,k\}$, let $r_i$ denote the unique vertex in $T^i$ such that
in $T$, $p(r_i)=r$. Choose $r_i$ as the root of $T^i$, for each $i\in\{1,2,\ldots,k\}$.
Observe that $T_e=(T^1_o\uplus T^2_o\uplus\cdots\uplus T^k_o)+T[\{r\}]$. By our earlier
observation, we know that each of $T^1_o,T^2_o,\ldots,T^k_o$ are cographs. Therefore,
it follows that $T_e$ is a cograph.
\end{proof}

\begin{Lemma}\label{lem:cographdimtree}
Let $T$ be a tree with a root. Then $T=T_o\cap T_e$.
\end{Lemma}
\begin{proof} 
Since each of $T_o$ and $T_e$ are supergraphs of $T$, we only need to show that $E(T_o)\cap E(T_e)
\subseteq E(T)$. Suppose for the sake of contradiction that there exists an edge
$xy\in E(T_o)\cap E(T_e)$ such that $xy\notin E(T)$. As $xy\in E(T_o)\cap E(T_e)$, we may assume without loss of generality that $y$ is a descendant of $p(x)$ in $T$. Since $xy\notin E(T)$, we have that $xy\notin E(T_o)$ if $x$ is even and $xy\notin E(T_e)$ if $x$ is odd. This contradicts our assumption that $xy\in E(T_o)\cap E(T_e)$. 
\end{proof}

Combining Lemma~\ref{lem:compiscograph} and Lemma~\ref{lem:cographdimtree}, we have the following theorem.
\begin{Theorem}
For any tree $T$, $\dim_{COG}(T)\leq 2$.
\end{Theorem}

Since cographs are closed under the disjoint union operation, we can now deduce from Observation~\ref{obs:join} that the cograph dimension of every forest is at most 2.
\begin{Corollary}\label{cor:cographdimforests}
Let $F$ be a forest then $\dim_{COG}(F)\leq 2$.
\end{Corollary}

Clearly, there are trees, even paths, that are not cographs, and therefore this bound is tight. Note that when $T$ is a path, then we can choose one of its endpoints as the root so that the graphs $T_o$ and $T_e$ are split graphs. Since $T_o$ and $T_e$ are also cographs by Lemma~\ref{lem:compiscograph}, we get that each of them is a threshold graph. We thus have the following from Lemma~\ref{lem:cographdimtree}.

\begin{Corollary}\label{cor:thresholddimpath}
For every path $P$, $\dim_{TH}(P)\leq 2$.
\end{Corollary}
\medskip

Next, we show that there are trees having threshold dimension at least 3.
Chv\'atal and Hammer~\cite{chvatalhammer} defined the auxiliary graph $G^*$ corresponding to a graph $G$ as follows: $V(G^*)=E(G)$ and two vertices $uv$ and $xy$ of $V(G^*)=E(G)$ are adjacent in $G^*$ if and only if $ux,vy\notin E(G)$. They asked whether $t(G)=\chi(G^*)$ for every graph $G$. Although Cozzens and Leibowitz~\cite{cozzens1984threshold} gave a negative answer to this question, 
Raschle and Simon~\cite{raschlesimon} proved the following theorem (a shorter proof was recently given in~\cite{dalu}).

\begin{Theorem}[Raschle-Simon]\label{thm:rashclesimon}
A graph $G$ has a threshold cover of size 2 if and only if $G^*$ is bipartite.
\end{Theorem}

This theorem can be used to prove that the tree $T$ shown in Figure~\ref{fig:thtree} has threshold dimension at least 3. Consider the graph $\overline{T}^*$. Note that $(bj)(ci)(dj)(ei)(fh)(gd)(fc)(eb)(da)(hb)(ia)(jb)$ is an odd cycle in $\overline{T}^*$. By Theorem~\ref{thm:rashclesimon}, we then have that $t(\overline{T})>2$. This implies that $\dim_{TH}(T)>2$.

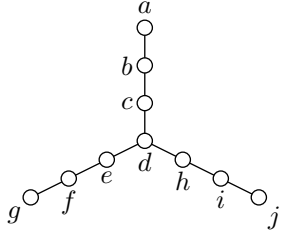
\begin{figure}
\begin{center}
\begin{tikzpicture}
\renewcommand{\vertexset}{(a,0,3),(b,0,2.5),(c,0,2),(d,0,1.5),(e,-.5,1.25),(f,-1,1),(g,-1.5,0.75),(h,.5,1.25),(i,1,1),(j,1.5,0.75)}
\renewcommand{\edgeset}{(a,b),(b,c),(c,d),(d,e),(e,f),(f,g),(d,h),(h,i),(i,j)}
\renewcommand{\defradius}{.1}
\drawgraph
\node[above=2] at (\xy{a}) {$a$};
\node[left=1] at (\xy{b}) {$b$};
\node[left=1] at (\xy{c}) {$c$};
\node[below=1] at (\xy{d}) {$d$};
\node[below=1] at (\xy{e}) {$e$};
\node[below=1] at (\xy{f}) {$f$};
\node[below left] at (\xy{g}) {$g$};
\node[below=1] at (\xy{h}) {$h$};
\node[below=1] at (\xy{i}) {$i$};
\node[below right] at (\xy{j}) {$j$};
\end{tikzpicture}
\caption{A tree that is not the intersection of two threshold graphs.}\label{fig:thtree}
\end{center}
\end{figure}

\section{Cycles}\label{sec:cycles}
We shall now turn our attention to cycles, and show that cycles on more than 6 vertices cannot be represented as the intersection of two cographs.

Let $C_n$ denote the cycle on $n$ vertices having $V(C_n)=\{v_1,v_2,\ldots,v_n\}$ and $E(C_n)=\{v_iv_{i+1}\colon 1\leq i\leq n-1\}\cup\{v_nv_1\}$. In the following, we shall let 
$v_{n+t}$ denote $v_t$ and $v_{1-t}$ denote $v_{n+1-t}$,
for $1\leq t\leq n$. We prove that for $n\ge 7$, there do not exist cographs $A$ 
and $B$ such that $C_n=A\cap B$. We will use the well-known fact that the diameter of any
induced subgraph of a cograph is at most 2; i.e. if $G'$ is an induced subgraph
of a cograph $G$, then there cannot be two vertices that are at a distance of 3
or more in $G'$ (since the shortest path between them in $G'$ will contain a $P_4$
that is an induced subgraph of $G'$, and therefore also an induced subgraph of $G$).

First we derive a lemma that must hold true for any cycle that is the intersection
of two cographs, provided that the cycle contains at least 6 vertices.

\begin{Lemma}\label{lem:3consecneighbour}
Let $A$ and $B$ be two cographs such that $C_n=A\cap B$ for some $n\geq 6$. 
Then there does not exist $i,j\in\{1,2,\ldots,n\}$ such that
$v_iv_j,v_iv_{j+1},v_iv_{j+2}\notin E(A)$.
\end{Lemma}
\begin{proof}
Suppose there exist $i,j\in\{1,2,\ldots,n\}$ such that
$v_iv_j,v_iv_{j+1},v_iv_{j+2}\notin E(A)$. Since $E(C_n)\subseteq E(A)$, we have
that $v_i \notin 
\{v_{j-1},v_j,v_{j+1},$ $v_{j+2},v_{j+3}\}$. Consider the graph $A$.
Let $v_x$ be the last neighbour of $v_i$ in the
sequence $v_{i+1},v_{i+2},\ldots,v_{j-1}$ and let $v_y$ be the first 
neighbour of $v_i$ in the sequence $v_{j+3},v_{j+4},\ldots,v_{i-1}$.
Clearly, the vertices $v_x,v_y,v_i$ are pairwise distinct.
We can see that 
$v_iv_xv_{x+1}\ldots v_jv_{j+1}v_{j+2}\ldots v_{y-1}v_yv_i$ 
forms a cycle with at least 6 vertices in the cograph $A$. Note that $v_x$ and 
$v_y$ are the only neighbours of $v_i$ in this cycle.
Consider the subgraph $A'$ induced in $A$ by the vertices in 
$\{v_{x+1},v_{x+2},\ldots,v_y,v_i\}$. Since $A'$ is clearly a connected induced
subgraph of the cograph $A$, there should be a path of length at most 
2 in $A'$ between $v_i$ and every other vertex in $A'$. As $v_y$ is the only 
neighbour of $v_i$ in $A'$, each vertex in $\{v_{x+1},v_{x+2},\ldots,v_{y-1}\}$ 
is adjacent to $v_y$ in $A'$ and therefore $v_yv_{x+1},v_yv_{x+2},\ldots,
v_yv_{y-1}\in E(A)$. By applying the same arguments on the subgraph induced in
$A$ by the vertices in $\{v_i,v_x,v_{x+1},\ldots,v_{y-1}\}$, we can show that
$v_xv_{x+1},v_xv_{x+2},\ldots,v_xv_{y-1}\in E(A)$. 

Now let us consider the graph $B$. As $C_n = A\cap B$, we can conclude from the
above observations that $v_xv_{x+2},v_xv_{x+3},\ldots,v_xv_{y-1},v_yv_{x+1},
v_yv_{x+2},\ldots,v_yv_{y-2}\notin E(B)$. Note that this means that there is no
vertex in $\{v_{x+1},v_{x+2},\ldots,v_{y-1}\}$ that is adjacent to both $v_x$
and $v_y$. Consider the subgraph $B'$ induced in $B$ by the vertices in $\{
v_x,v_{x+1},\ldots,v_y\}$. As $B'$ is a connected 
induced subgraph of the cograph $B$, there must be a path of distance at most 2 
between $v_x$ and $v_y$ in $B'$. Since there is no
vertex in $\{v_{x+1},v_{x+2},\ldots,v_y\}$ that is adjacent to both $v_x$ and
$v_y$, this implies that $v_xv_y\in E(B')$, and therefore $v_xv_y\in E(B)$.
But then, $v_xv_yv_{y-1}v_{y-2}$ is an induced $P_4$ in $B$, contradicting the
fact that $B$ is a cograph.
\end{proof}

\begin{Lemma}\label{lem:2consecneighbour}
Let $A$ and $B$ be two cographs such that $C_n=A\cap B$ for some $n\geq 7$. 
Then there does not exist $i,j\in\{1,2,\ldots,n\}$ such that $v_iv_j,v_iv_{j+1}
\notin E(A)$ and $v_i\notin\{v_{j-2},v_{j+3}\}$.
\end{Lemma}
\begin{proof}
Let us assume that such $i$ and $j$ exist. Note that since $v_iv_j,v_iv_{j+1}
\notin E(A)\supseteq E(C_n)$ and $v_i\notin\{v_{j-2},v_{j+3}\}$, we can
conclude that $v_i\notin\{v_{j-2},v_{j-1},v_j,v_{j+1},v_{j+2},v_{j+3}\}$.
By Lemma~\ref{lem:3consecneighbour},
$v_iv_{j-1}$ and $v_iv_{j+2} \in E(A)$. Note that $v_jv_{j+2}\in 
E(A)$, as otherwise $v_iv_{j+2}v_{j+1}v_j$ will be an induced $P_4$ in $A$. 
Similarly, $v_{j-1}v_{j+1}\in E(A)$, as otherwise $v_iv_{j-1}v_jv_{j+1}$ will 
form an induced $P_4$ in $A$. As $C_n=A\cap B$, we now have that $v_iv_{j-1},
v_iv_{j+2},v_{j-1}v_{j+1},v_jv_{j+2}\notin E(B)$.
This tells us that $v_{j-1}v_{j+2}\in E(B)$, as otherwise, 
$v_{j-1}v_jv_{j+1}v_{j+2}$ will be an induced $P_4$ in $B$. Further, using
Lemma~\ref{lem:3consecneighbour} on $B$, we can infer that at least one of 
the edges $v_iv_j,v_iv_{j+1}$ must be present in $E(B)$.
If $v_iv_j\in E(B)$, then $v_iv_jv_{j-1}v_{j+2}$ will be an induced $P_4$ in 
$B$ 
and if $v_iv_{j+1}\in E(B)$, then $v_iv_{j+1}v_{j+2}v_{j-1}$ will be an induced 
$P_4$ in $B$. In either case, we have a contradiction to the fact that $B$ is a 
cograph.
\end{proof}

\begin{Corollary}\label{Cor:edgesincographs}
Let $C_n=A\cap B$, for some $n\geq 7$, where $A$ and $B$ are cographs, and
let $i\in\{1,2,\ldots,n\}$. Then:
\begin{itemize}
\item exactly one of the edges in $\{v_iv_{i+3},
v_iv_{i+4}\}$ is in $E(A)$ and the other is in $E(B)$, and
\item exactly one
of the edges in $\{v_iv_{i-3},v_iv_{i-4}\}$ is in $E(A)$ and the other is
in $E(B)$.
\end{itemize}
\end{Corollary}
\begin{proof}
By Lemma~\ref{lem:2consecneighbour}, at least one of $v_iv_{i+3},v_iv_{i+4}$ is
in $E(A)$. Similarly, by applying Lemma~\ref{lem:2consecneighbour} on $B$,
at least one of these two edges is also in $E(B)$.
Since $v_iv_{i+3},v_iv_{i+4}\notin E(C_n)= E(A)\cap E(B)$, we can
conclude that exactly one of them is in $A$ and the other in $B$. By symmetry,
we have that exactly one of the edges in $\{v_iv_{i-3},v_iv_{i-4}\}$ is in
$E(A)$ and the other is in $E(B)$.
\end{proof}
\begin{Lemma}\label{lem:edgesinA}
Let $C_n=A\cap B$, for some $n\geq 7$, where $A$ and $B$ are cographs. If for
some $i\in\{1,2,\ldots,n\}$, the edge $v_iv_{i+4}\in E(A)$, then
$v_{i+1}v_{i+5}$ is in $E(A)$.
\end{Lemma}
\begin{proof}
Suppose that $v_iv_{i+4}\in E(A)$. From Corollary~\ref{Cor:edgesincographs}, exactly one of the edges in 
$\{v_{i+1}v_{i+4},v_{i+1}v_{i+5}\} \in E(A)$. If $v_{i+1}v_{i+4} 
\in E(A)$, then the vertex $v_{i+4}$ is adjacent to both the vertices 
$v_{(i+4)-3}$ and $v_{(i+4)-4}$ in $A$, which contradicts 
Corollary~\ref{Cor:edgesincographs}. Thus, $v_{i+1}v_{i+4} \notin E(A)$, which 
by 
Corollary~\ref{Cor:edgesincographs} means that $v_{i+1}v_{i+5}\in E(A)$.
\end{proof}

\begin{Theorem}\label{Thm:cographdimcycle}
There does not exist two cographs $A$ and $B$ such that $A \cap B$ = $C_n$, when 
$n \geq 7$. In other words, $\dim_{COG}(C_n)>2$, when $n\geq 7$.
\end{Theorem}
\begin{proof}
Suppose such cographs $A$ and $B$ exist. From
Corollary~\ref{Cor:edgesincographs}, we may assume 
without loss of generality that $v_1v_5\in E(A)$. Then by applying
Lemma~\ref{lem:edgesinA} repeatedly, we can conclude that for every $i\in\{1,2,
\ldots,n\}$, $v_iv_{i+4}\in E(A)$. By Corollary~\ref{Cor:edgesincographs}, this
implies that for every $i\in\{1,2,\ldots,n\}$, $v_iv_{i+3}\notin E(A)$.
If $v_iv_{i+2}\notin E(A)$ for some $i\in\{1,2,\ldots,n\}$, 
then $v_iv_{i-1}v_{i+3}v_{i+2}$ will form an induced $P_4$ in $A$, 
contradicting 
the fact that $A$ is a cograph. Therefore, we can conclude that for every 
$i\in\{1,2,\ldots,n\}$, $v_iv_{i+2}\in E(A)$. Let us consider the path 
$v_1v_3v_4v_6$ in $A$. Since $v_1v_5\in E(A)$, we have $v_1v_5\notin E(B)$. If
$n\geq 8$, then we can apply Lemma~\ref{lem:2consecneighbour} on $B$ to get
that $v_1v_6\in E(B)$, which implies that $v_1v_6\notin E(A)$. Thus, if $n\geq
8$, we have the induced $P_4$ $v_1v_3v_4v_6$ in $A$, which contradicts the fact
that $A$ is a cograph. We can therefore conclude that $n=7$. Then since $v_iv_{i+4}
\in E(A)$ for each $i\in\{1,2,\ldots,n\}$, we have $v_1v_5,v_4v_1\in E(A)$. This
contradicts Corollary~\ref{Cor:edgesincographs}.
\end{proof}

Figure~\ref{fig:smallcycles} shows that every cycle on less than 7 vertices is the
intersection of at most two cographs (note that cycles on less than 5 vertices are
themselves cographs).

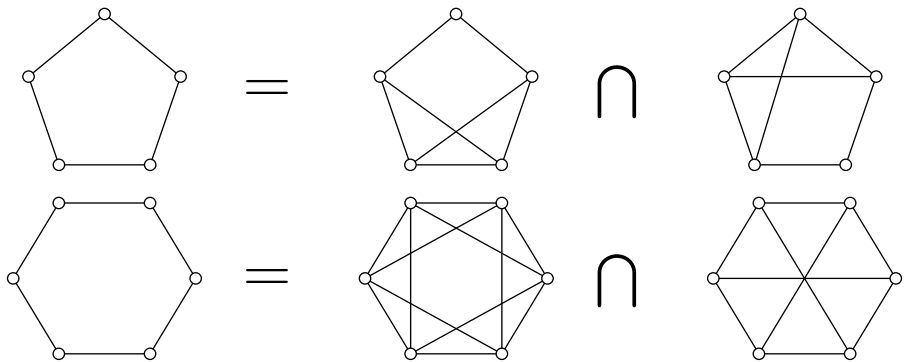
\begin{figure}[h]
\begin{center}
\begin{tabular}{ccccc}
\renewcommand{\vertexset}{(x1,0.9,0),(x2,0.5,1.17),(x3,1.5,2),(x4,2.5,1.17),(x5,2.1,0)}
\renewcommand{\edgeset}{(x1,x2),(x2,x3),(x3,x4),(x4,x5),(x5,x1)}
\renewcommand{\defradius}{.075}
\begin{tikzpicture}
\drawgraph
\end{tikzpicture}
&
\begin{tikzpicture}[baseline=-30]
\node [thick,font=\fontsize{60}{0}\selectfont] at (0,0) {$=$};
\end{tikzpicture}
&
\renewcommand{\vertexset}{(x1,0.9,0),(x2,0.5,1.17),(x3,1.5,2),(x4,2.5,1.17),(x5,2.1,0)}
\renewcommand{\edgeset}{(x1,x2),(x2,x3),(x3,x4),(x4,x5),(x5,x1),(x1,x4),(x2,x5)}
\renewcommand{\defradius}{.075}
\begin{tikzpicture}
\drawgraph
\end{tikzpicture}
&
\begin{tikzpicture}[baseline=-30]
\node [font=\fontsize{20}{0}] at (0,0) {$\bigcap$};
\end{tikzpicture}
&
\renewcommand{\vertexset}{(x1,0.9,0),(x2,0.5,1.17),(x3,1.5,2),(x4,2.5,1.17),(x5,2.1,0)}
\renewcommand{\edgeset}{(x1,x2),(x2,x3),(x3,x4),(x4,x5),(x5,x1),(x1,x3),(x2,x4)}
\renewcommand{\defradius}{.075}
\begin{tikzpicture}
\drawgraph
\end{tikzpicture}
\medskip\\
\renewcommand{\vertexset}{(x1,1,0),(x2,0.4,1),(x3,1,2),(x4,2.2,2),(x5,2.8,1),(x6,2.2,0)}
\renewcommand{\edgeset}{(x1,x2),(x2,x3),(x3,x4),(x4,x5),(x5,x6),(x6,x1)}
\renewcommand{\defradius}{.075}
\begin{tikzpicture}
\drawgraph
\end{tikzpicture}
&
\begin{tikzpicture}[baseline=-30]
\node [thick,font=\fontsize{60}{0}\selectfont] at (0,0) {$=$};
\end{tikzpicture}
&
\renewcommand{\vertexset}{(x1,1,0),(x2,0.4,1),(x3,1,2),(x4,2.2,2),(x5,2.8,1),(x6,2.2,0)}
\renewcommand{\edgeset}{(x1,x2),(x2,x3),(x3,x4),(x4,x5),(x5,x6),(x6,x1),(x1,x3),(x2,x4),(x3,x5),(x4,x6),(x5,x1),(x6,x2)}
\renewcommand{\defradius}{.075}
\begin{tikzpicture}
\drawgraph
\end{tikzpicture}
&
\begin{tikzpicture}[baseline=-30]
\node [font=\fontsize{20}{0}] at (0,0) {$\bigcap$};
\end{tikzpicture}
&
\renewcommand{\vertexset}{(x1,1,0),(x2,0.4,1),(x3,1,2),(x4,2.2,2),(x5,2.8,1),(x6,2.2,0)}
\renewcommand{\edgeset}{(x1,x2),(x2,x3),(x3,x4),(x4,x5),(x5,x6),(x6,x1),(x1,x4),(x2,x5),(x3,x6)}
\renewcommand{\defradius}{.075}
\begin{tikzpicture}
\drawgraph
\end{tikzpicture}
\end{tabular}
\end{center}
\caption{Cycles with less than 7 vertices are the intersection of at most two cographs.}
\label{fig:smallcycles}
\end{figure}

\begin{Remark}\label{rem:dimcographleq6}
$\dim_{COG}(C_n)\leq 2$, when $n\leq 6$.	
\end{Remark}

It is easy to see that any cycle $C$ is the intersection of at most three threshold graphs as follows. Let $v$ be any vertex on the cycle. Then $C-v$ is a path and therefore, by Corollary~\ref{cor:thresholddimpath}, there exist two threshold graphs $G_1$ and $G_2$ such that $G_1\cap G_2=C-v$. Let $G'_1$ and $G'_2$ be obtained from $G_1$ and $G_2$ respectively by adding $v$ as a universal vertex. Since the graph obtained by adding a universal vertex to any threshold graph is again a threshold graph, $G'_1$ and $G'_2$ are threshold graphs. Let $G_3$ be the graph obtained from $C$ by making every pair of vertices from $V(C)\setminus\{v\}$ adjacent to each other (thus, $V(C)\setminus\{v\}$ is a clique in $G_3$). It is easy to verify that $G_3$ is a threshold graph and that $C=G'_1\cap G'_2\cap G_3$. Therefore, we have the following result.

\begin{Remark}\label{rem:dimcographcycle}
For any cycle $C$, $\dim_{COG}(C)\leq\dim_{TH}(C)\leq 3$.
\end{Remark}

It is easy to see that a $C_4$ is a cograph but not a threshold graph. However, it can be represented as the intersection of 2 threshold graphs as shown in Figure~\ref{fig:c4threshold}.
The graph $C_5$ can be seen to have threshold dimension at least 3 as follows. Suppose for the sake of contradiction that $\dim_{TH}(C_5)\leq 2$. Since $\overline{C_5}$ is isomorphic to $C_5$, we have $t(C_5)=t(\overline{C_5})\leq 2$, which means that the edges of a $C_5$ can be covered by at most two threshold graphs. Let $H_1$ and $H_2$ be two threshold graphs that cover the edges of a $C_5$. If $H_i$, for some $i\in\{1,2\}$, contains three edges, then it contains either an induced $P_4$ or an induced $2K_2$, which is a contradiction to the fact that $H_i$ is a threshold graph. So each of $H_1$ and $H_2$ can cover at most two edges of the $C_5$, which contradicts the fact that every edge of the $C_5$ is contained in at least one of $H_1$ or $H_2$.

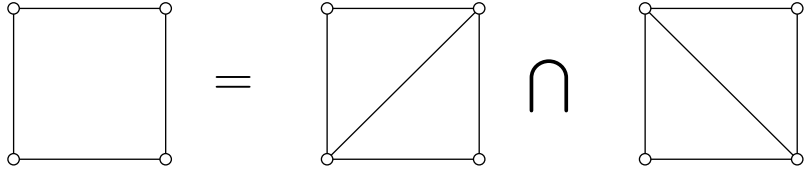
\begin{figure}[h]
\begin{center}
\begin{tabular}{ccccc}
\renewcommand{\vertexset}{(x1,0,0),(x2,0,2),(x3,2,2),(x4,2,0)}
\renewcommand{\edgeset}{(x1,x2),(x2,x3),(x3,x4),(x1,x4)}
\renewcommand{\defradius}{.075}
\begin{tikzpicture}
\drawgraph
\end{tikzpicture}
&
\begin{tikzpicture}[baseline=-30]
\node [font=\fontsize{20}{0}] at (0,0) {$=$};
\end{tikzpicture}
&
\renewcommand{\vertexset}{(x1,0,0),(x2,0,2),(x3,2,2),(x4,2,0)}
\renewcommand{\edgeset}{(x1,x2),(x2,x3),(x3,x4),(x1,x4),(x1,x3)}
\renewcommand{\defradius}{.075}
\begin{tikzpicture}
\drawgraph
\end{tikzpicture}
&
\begin{tikzpicture}[baseline=-30]
\node [font=\fontsize{20}{0}] at (0,0) {$\bigcap$};
\end{tikzpicture}
&
\renewcommand{\vertexset}{(x1,0,0),(x2,0,2),(x3,2,2),(x4,2,0)}
\renewcommand{\edgeset}{(x1,x2),(x2,x3),(x3,x4),(x1,x4),(x2,x4)}
\renewcommand{\defradius}{.075}
\begin{tikzpicture}
\drawgraph
\end{tikzpicture}
\end{tabular}
\end{center}
\caption{The threshold dimension of $C_4$ is 2.}
\label{fig:c4threshold}
\end{figure}

We can see that the threshold dimension of a $C_6$ is also at least 3 as follows. Suppose for the sake of contradiction that $\dim_{TH}(C_6)\leq 2$. Consider the cycle $C=v_1v_2\ldots v_6v_1$ on 6 vertices. Let $G=\overline{C}$. From our assumption, we have $t(G)\leq 2$. Thus there exist two threshold graphs $H_1$ and $H_2$ such that $E(G)=E(H_1)\cup E(H_2)$. Now consider the edges $v_1v_4$, $v_2v_5$ and $v_3v_6$ in $E(G)$. If any two of them belong to $H_i$, for some $i\in\{1,2\}$, then $H_i$ would contain an induced $2K_2$, $P_4$ or $C_4$, which is a contradiction (to see this, note that for any two of these edges, their endpoints induce a $C_4$ in $G$). Therefore at least one of the edges $v_1v_4,v_2v_5,v_3v_6$ is not contained in either $H_1$ or $H_2$, contradicting the fact that $E(G)=E(H_1)\cup E(H_2)$.

Note that we could also have used Theorem~\ref{thm:rashclesimon} to show that $C_5$ and $C_6$ are not the intersection of two threshold graphs. We thus get the values shown in Table~\ref{tab:cycledim} for the cograph dimension and threshold dimension of cycles of every size.

\begin{table}[h]
\centering
 \begin{tabular}{|c|c|c|}
 \hline
 $n$ & $\dim_{COG}(C_n)$ & $\dim_{TH}(C_n)$ \\
 \hline
 3 & 1 & 1 \\
 \hline
 4 & 1 & 2 \\
 \hline
 5 & 2 & 3 \\
 \hline
 6 & 2 & 3 \\
 \hline
 $\geq 7$ & 3 & 3\\
 \hline
 \end{tabular}
 \caption{Cograph and threshold dimensions of cycles}
 \label{tab:cycledim}
 \end{table}

\section{Threshold dimension, boxicity and chromatic number}\label{sec:threshboxichro}

A $k$-dimensional box, or \emph{$k$-box} for short, is defined as the Cartesian product of $k$ closed intervals in $\mathbb{R}$. The \emph{boxicity} of a graph $G$, denoted by $\boxi(G)$, is the minimum integer $k$ for which there is an assignment of $k$-boxes to the vertices of $G$ such that for $u,v\in V(G)$, $uv\in E(G)$ if and only if the $k$-boxes corresponding to $u$ and $v$ intersect. It is known that for any graph $G$, $\boxi(G)=\dim_{INT}(G)$, where $INT$ denotes the class of interval graphs~\cite{cozzens1989dimensional}. The \emph{chromatic number} of a graph $G$, denoted as $\chi(G)$, is the minimum number of colours required to colour the vertices of $G$ such that no two adjacent vertices receive the same colour.


We shall now prove the main result of this section.

\begin{Theorem}\label{thm:boxi}
	Let $G$ be any graph. Then, $\dim_{TH}(G)\leq \chi(G)\boxi(G)$.
\end{Theorem}
\begin{proof}
	Let $\boxi(G)=k$. Then there exists a collection of $k$-boxes $\{B_u\}_{u\in V(G)}$ such that for $u,v\in V(G)$, $uv\in E(G)\Leftrightarrow B_u\cap B_v\neq\emptyset$. For $u\in V(G)$, let $B_u=[l_1(u),r_1(u)]\times [l_2(u),r_2(u)]\times\cdots\times [l_k(u),r_k(u)]$. Note that this means that $uv\notin E(G)$ if and only if $\exists j\in\{1,2,\ldots,k\}$ such that $r_j(u)<l_j(v)$ or $r_j(v)<l_j(u)$.
	
	Let $c:V(G)\rightarrow\{1,2,\ldots,\chi(G)\}$ be a proper vertex colouring of $G$.
	
	Now, for each $i\in\{1,2,\ldots,\chi(G)\}$ and $j\in\{1,2,\ldots,k\}$, we construct a graph $G_{ij}$ on vertex set $V(G)$ as follows.
	Define $E(G_{ij})=\{uv~|~c(u)\neq i,c(v)\neq i\}\cup\{uv~|~c(u)=i,c(v)\neq i$ and $r_j(u)\geq l_j(v)\}$. Note that each $G_{ij}$ is a split graph (the vertices with colour $i$ form an independent set and all other vertices form a clique). We shall now show that $$G=\bigcap_{\substack{1\leq i\leq\chi(G)\\1\leq j\leq k}} G_{ij}$$
	If $uv\in E(G)$, then clearly $c(u)\neq c(v)$ and for each $j\in\{1,2,\ldots,k\}$, we have $r_j(u)\geq l_j(v)$ and $r_j(v)\geq l_j(u)$. It follows that each $G_{ij}$ is a supergraph of $G$. So we only need to show that for each $u,v\in V(G)$ such that $uv\notin E(G)$, there exists $i\in\{1,2,\ldots,\chi(G)\}$ and $j\in\{1,2,\ldots,k\}$ such that $uv\notin E(G_{ij})$. Consider $uv\notin E(G)$ (where $u\neq v$). If $c(u)=c(v)=i$, then clearly, $uv\notin E(G_{ij})$, for all $1\leq j\leq k$. So let us assume that $c(u)\neq c(v)$. As $uv\notin E(G)$, we can assume without loss of generality that there exists $j\in\{1,2,\ldots,k\}$ such that $r_j(u)<l_j(v)$. Let $c(u)=i$, which means that $c(v)\neq i$. It now follows from the definition of $G_{ij}$ that $uv\notin E(G_{ij})$.

	To complete the proof, we shall show that each $G_{ij}$ is a cograph, which would imply that it is also a threshold graph, as every split graph that is also a cograph is a threshold graph~\cite{Brandstadt:1999:GCS:302970}. In particular, we show that no $G_{ij}$ contains an induced $P_4$. Suppose for the sake of contradiction that there exists $i\in\{1,2,\ldots,\chi(G)\}$ and $j\in\{1,2,\ldots,k\}$ such that $G_{ij}$ contains an induced path $wxyz$. As $\{u\in V(G_{ij})~|~c(u)=i\}$ forms an independent set in $G_{ij}$ and $\{u\in V(G_{ij})~|~c(u)\neq i\}$ forms a clique in $G_{ij}$, we can conclude that $c(w)=c(z)=i$, $c(x)\neq i$ and $c(y)\neq i$. Since $wx\in E(G_{ij})$ and $wy\notin E(G_{ij})$, we get $l_j(y)>r_j(w)\geq l_j(x)$. Similarly, since $yz\in E(G_{ij})$ and $xz\notin E(G_{ij})$, we get $l_j(x)>r_j(z)\geq l_j(y)$. We now have both $l_j(y)>l_j(x)$ and $l_j(x)>l_j(y)$, which is a contradiction.
\end{proof}

Since threshold graphs form a subclass of cographs, we have the following corollary.

\begin{Corollary}
Let $G$ be any graph. Then, $\dim_{COG}(G)\leq \dim_{TH}(G)\leq \chi(G)\boxi(G)$.
\end{Corollary}

A \emph{grid intersection graph} is a graph whose vertices can be mapped to horizontal and vertical line segments in the plane in such a way that two vertices are adjacent in the graph if and only if the line segments corresponding to them intersect and no two parallel line segments intersect. Clearly, every grid intersection graph has boxicity at most 2, since horizontal and vertical line segments are just degenerate 2-boxes. We shall use the following theorem of Hartman et al.~\cite{hartman1991grid}.

\begin{Theorem}[Hartman et al.]\label{thm:Hartman}
Every planar bipartite graph is a grid intersection graph.
\end{Theorem}

From Theorem~\ref{thm:Hartman} and Theorem~\ref{thm:boxi}, we now have the following corollary.

\begin{Corollary}\label{cor:plbipthdim}
For any planar bipartite graph $G$, $\dim_{COG}(G)\leq \dim_{TH}(G)\leq 4$.
\end{Corollary}

Note that the above corollary is a strengthening of Theorem~3.10 of~\cite{kratochvil1994intersection} which states that for every planar bipartite graph $G$, $\dim_{PER}(G)\leq 4$, where $PER$ denotes the class of permutation graphs (note that cographs, and therefore threshold graphs, form a subclass of permutation graphs).
Using the Four Colour Theorem~\cite{fct} and the fact that the boxicity of any planar graph is at most 3~\cite{thomassen1986interval}, Theorem~\ref{thm:boxi} also implies the following strengthening of Corollary~3.11 of~\cite{kratochvil1994intersection} which states that $\dim_{PER}(G)\leq 12$ for every planar graph $G$.

\begin{Corollary}\label{cor:thresholdplanar}
For any planar graph $G$, $\dim_{TH}(G)\leq 12$.
\end{Corollary}

Since every triangle-free planar graph has a 3-colouring (Gr\"otzsch's Theorem), we also get the following.
\begin{Corollary}\label{cor:thresholdtrianglefreeplanar}
For any triangle-free planar graph $G$, $\dim_{TH}(G)\leq 9$.
\end{Corollary}

Since for any outerplanar graph $G$, we have $\chi(G)\leq 3$ (folklore) and
$\boxi(G)\leq 2$~\cite{scheinerman}, we get the following result as a corollary of Theorem~\ref{thm:boxi}. 

\begin{Corollary}\label{cor:thresholdouterplanar}
For any outerplanar graph $G$, $\dim_{TH}(G)\leq 6$.
\end{Corollary}

Partial 2-trees are also 3-colourable, but there exist partial 2-trees that have boxicity~3~\cite{chandranseries}. Thus we have the following.

\begin{Corollary}\label{cor:thresholdp2tree}
For any partial 2-tree $G$, $\dim_{TH}(G)\leq 9$.
\end{Corollary}

Adiga, Bhowmick and Chandran~\cite{adiga} showed that the boxicity of any graph with maximum degree $\Delta$ is $O(\Delta\log^2\Delta)$. This means, by Theorem~\ref{thm:boxi} and the fact that $\chi(G)\leq\Delta+1$, that $\dim_{TH}(G)$ is $O(\Delta^2\log^2\Delta)$ for any graph $G$ with maximum degree $\Delta$. 
\section{Upper bounds using treewidth and pathwidth}\label{sec:cogdimtwd}
A tree decomposition of a graph $G$ is a pair $(T,f)$, where $f:V(T)\rightarrow 2^{V(G)}$ is an assignment of subsets of $V(G)$ to the vertices of a tree $T$, satisfying the following properties:
\begin{enumerate}
\item\label{prop:conn} For every vertex $u\in V(G)$, the subgraph induced in $T$ by the set $\{x\in V(T)\colon u\in f(x)\}$ is connected (note that this implies that every vertex $u\in V(G)$ is contained in $f(x)$ for some $x\in V(T)$). Observe that this means that if $u\in f(x)\cap f(y)$, for some $x,y\in V(T)$, then for every $z\in V(T)$ that lies on the path between $x$ and $y$ in $T$, $u\in f(z)$.
\item\label{prop:edge} For every edge $uv\in E(G)$, there exists $x\in V(T)$ such that $u,v\in f(x)$.
\end{enumerate}
The width of a tree decomposition $(T,f)$ of a graph $G$ is defined as $\max_{x\in V(T)}\{|f(x)|-1\}$. The \emph{treewidth} of $G$, denoted by $\tw(G)$, the minimum width of a tree decomposition of $G$.
\medskip

Chandran and Sivadasan~\cite{chandran2007boxicity} showed that for any graph $G$, $\boxi(G)\leq\tw(G)+2$. Combining this with Theorem~\ref{thm:boxi} and the fact that $\chi(G)\leq\tw(G)+1$ (folklore), we get that for every graph $G$, $\dim_{TH}(G)\leq\chi(G)(\tw(G)+2)\leq(\tw(G)+1)(\tw(G)+2)$.
As shown in~\cite{aravindcrs}, the fact that the star chromatic number of a graph $G$ is at most $\frac{(\tw(G)+1)(\tw(G)+2)}{2}$~\cite{fertin2004star} can be combined with Theorem~\ref{thm:star} from Section~\ref{sec:vertpart} to deduce that $\dim_{COG}(G)\leq\frac{(\tw(G)+1)(\tw(G)+2)}{2}$ for any graph $G$. In this section we show that this upper bound can be improved.

We would like to note here that though the general strategy used in the lemma below is similar to that in the proof of Theorem~14 in~\cite{chandran2007boxicity}, the details in the two proofs are very different.
\begin{Lemma}\label{lem:twpw}
Let $(T,f)$ be a tree decomposition of a graph $G$ having width $k-1$. Then there exist cographs $G_0,G_1,\ldots,G_k$ such that $G=G_0\cap G_1\cap\cdots\cap G_k$.
\end{Lemma}
\begin{proof}
As $(T,f)$ is a tree decomposition of $G$ having width $k-1$, we have that for each $x\in V(T)$, $|f(x)|\leq k$.

Let $G'$ be the graph with $V(G')=V(G)$ and $E(G')=\{uv\colon\exists x\in V(T)$ such that $u,v\in f(x)\}$. Note that $G'$ is a supergraph of $G$ and also that $G'$ is a chordal graph (folklore; to see this, note that the intersection graph of subtrees of a tree is a chordal graph~\cite{gavril}). Note that $(T,f)$ is a tree decomposition of $G'$ as well. Since $G'$ is a chordal graph, $\chi(G')=\omega(G')=\tw(G')+1$~\cite{golumbic,graphminors2}. Therefore, there exists a proper vertex colouring $c:V(G')\rightarrow\{1,2,\ldots,k\}$ of $G'$. It is easy to see that $c$ is proper vertex colouring of $G$ too. In particular, we have the stronger property that if $u,v\in V(G)$ such that there exists $x\in V(T)$ having $u,v\in f(x)$, then $c(u)\neq c(v)$.

Choose an arbitrary vertex $r\in V(T)$ as the root of $T$ and define the ancestor-descendant relation among the vertices of $T$ in the usual way (i.e., for $x,y\in V(T)$, $x$ is an ancestor of $y$ if and only if $x$ lies on the path between $r$ and $y$ in $T$).
For every vertex $u\in V(G)$, define $h(u)$ to be that vertex in $f^{-1}(u)=\{x\in V(T)\colon u\in f(x)\}$ that is an ancestor of every other vertex in $f^{-1}(u)$ (it is easy to see, using property~\ref{prop:conn} of tree decompositions, that there is always exactly one such vertex). We shall now construct a binary relation $R$ on $V(G)$ as follows. Let $<$ be an arbitrarily chosen linear ordering of the vertices in $V(G)$. For two vertices $u,v\in V(G)$ such that $h(u)\neq h(v)$, we have $uRv$ if and only if $h(u)$ is an ancestor of $h(v)$ in $T$. If $h(u)=h(v)$, then we have $uRv$ if and only if $u<v$. Note that $R$ is a partial order on $V(G)$.

We now construct the $k+1$ cographs $G_0,G_1,\ldots,G_k$ such that $G=\bigcap_{i=0}^k G_i$ as follows.

Define $V(G_0)=V(G)$ and $E(G_0)=\{uv\colon uRv\}$. By property~\ref{prop:edge} of tree decompositions, we have that for any edge $uv\in E(G)$, one of $h(u)$ or $h(v)$ is an ancestor of the other (this includes the case $h(u)=h(v)$). Therefore, $G_0$ is a supergraph of $G$.

For each $i\in\{1,2,\ldots,k\}$, define $V(G_i)=V(G)$ and $E(G_i)=E(G)\cup\{uv\colon (u,v),(v,u)\notin R\}\cup\{uv\colon uRv$ and $c(u)\neq i\}$. Clearly, $G_i$ is a supergraph of $G$.

First, we shall show that $G=\bigcap_{i=0}^k G_i$. For this, we only need to show that for any $uv\notin E(G)$, there exists $i\in\{0,1,\ldots,k\}$ such that $uv\notin E(G_i)$. Let $uv\notin E(G)$. If $(u,v),(v,u)\notin R$, then $uv\notin E(G_0)$. So let us assume without loss of generality that $uRv$. Then from the definition of $G_{c(u)}$, it is clear that $uv\notin E(G_{c(u)})$.

It only remains to be proven that each $G_i$, $0\leq i\leq k$, is a cograph.

Suppose for the sake of contradiction that $G_0$ is not a cograph. Then there exists an induced path $P=abcd$ in $G_0$. From the definition of $G_0$, $uv\in E(G_0)$ if and only if either $uRv$ or $vRu$. Let us orient the edges of $P$ such that an edge $uv$ of $P$ gets oriented from $u$ to $v$ if and only if $uRv$. Let $\hat{P}$ be the path $P$ together with the orientations on its edges. Suppose that there is a directed path of length 2 in $\hat{P}$, which we will assume without loss of generality to be $a\rightarrow b\rightarrow c$. Recalling that $R$ is a partial order, $aRb$ and $bRc$ together implies $aRc$, which contradicts the fact that $ac\notin E(G_0)$. Thus, there is not directed path of length 2 in $\hat{P}$. We can therefore assume without loss of generality that $\hat{P}$ is the path $a\rightarrow b\leftarrow c\rightarrow d$. We then have both $aRb$ and $cRb$, which implies by the definition of $R$ that either $aRc$ or $cRa$. But then $ac\in E(G_0)$, which is a contradiction.

Now let us suppose for the sake of contradiction that $G_i$ is not a cograph, for some $i\in\{1,2,\ldots,k\}$. Then there exists an induced path $P=abcd$ in $G_i$. Let $Q$ be the path $bdac$ in $\overline{G_i}$. From the definition of $G_i$, it is clear that if $uv\in E(\overline{G_i})$, then either $uRv$ and $c(u)=i$ or $vRu$ and $c(v)=i$. Now let us orient the edges of $Q$ such that an edge $uv\in E(Q)$ gets oriented from $u$ to $v$ if and only if $uRv$. Let $\hat{Q}$ be the path $Q$ together with the orientations on its edges. Suppose that there is a directed path of length 2 in $\hat{Q}$, which we will assume without loss of generality to be $b\rightarrow d\rightarrow a$. By our observation above, we have that $bRd$, $dRa$, and $c(b)=c(d)=i$. This means that $bRa$ and that $h(d)$ lies on the path between $h(b)$ and $h(a)$ in $T$. Since $ab\in E(G)$, the former and the definition of $h$ together implies that $b\in f(h(a))$. By property~\ref{prop:conn} of tree decompositions, the latter now implies that $b\in f(h(d))$. Since $d\in f(h(d))$, this contradicts the fact that $c(b)=c(d)$. Therefore, we can assume that there is no directed path of length 2 in $\hat{Q}$. Then we can assume without loss of generality that $\hat{Q}$ is $b\rightarrow d\leftarrow a\rightarrow c$. This means that we have $bRd$, $aRd$, and $c(b)=c(a)=i$. By the definition of $R$, we then have that either $bRa$ or $aRb$. Since $c(a)=c(b)$, we have $ab\notin E(G)$. As $c(a)=c(b)=i$, we further have $ab\notin E(G_i)$, which is a contradiction.
\end{proof}

\begin{Theorem}\label{thm:treewidth}
For any graph $G$, $\dim_{COG}(G)\leq\tw(G)+2$.
\end{Theorem}
\begin{proof}
Let $(T,f)$ be a tree decomposition of $G$ having width $\tw(G)$. Then by Lemma~\ref{lem:twpw}, there exist $\tw(G)+2$ cographs $G_0,G_1,\ldots,G_{\tw(G)+1}$ such that $G=G_0\cap G_1\cap\cdots\cap G_{\tw(G)+1}$. This proves the theorem.
\end{proof}

\begin{Corollary}\label{cor:cogdimp2tree}
For every partial 2-tree $G$, $\dim_{COG}(G)\leq 4$.
\end{Corollary}

Since outerplanar graphs are partial 2-trees, this means that every outerplanar graph is the intersection of at most 4 cographs.


A tree decomposition $(T,f)$ of a graph $G$ is said to be a \emph{path decomposition} of $G$ if $T$ is a path. The pathwidth of $G$, denoted by $\pw(G)$, is defined as the minimum width of a path decomposition of $G$. Clearly, for any graph $G$, $\tw(G)\leq\pw(G)$.

\begin{Theorem}\label{thm:pathwidth}
For every graph $G$, $\dim_{TH}(G)\leq\pw(G)+1$.
\end{Theorem}
\begin{proof}
Let $(T,f)$ be a path decomposition of $G$ of width $\pw(G)$. Following the proof of Lemma~\ref{lem:twpw}, we can select an endvertex of the path $T$ to be the root $r$ and construct the graphs $G_0,G_1,\ldots,G_{\pw(G)+1}$. Then, the relation $R$ defined on $V(G)$ has the property that for any two vertices $u,v\in V(G)$, we have either $uRv$ or $vRu$. Therefore, the graphs $G_0,G_1,\ldots,$ $G_{\pw(G)+1}$ have the following properties:
\begin{enumerate}
\vspace{-0.05in}
\itemsep 0in
\renewcommand{\labelenumi}{(\alph{enumi})}
\item[(a)] $G_0$ is a complete graph, and
\item[(b)] for each $i\in\{1,2,\ldots,\pw(G)+1\}$, $G_i$ is a split graph (the vertices in $\{u\colon c(u)=i\}$ form an independent set and the rest form a clique in $G_i$), and therefore a threshold graph.
\end{enumerate}
From (a), we have that $G=G_1\cap G_2\cap\cdots\cap G_{\pw(G)+1}$ and so by (b), we conclude that the graph $G$ can be represented as the intersection of $\pw(G)+1$ threshold graphs. Thus, we have the theorem.
\end{proof}

Although Corollary~\ref{cor:thresholdouterplanar} says that every outerplanar graph is the intersection of at most 6 threshold graphs, we do not know of any outerplanar graphs that have threshold dimension more than 3 (cycles of more than 6 vertices have threshold dimension equal to 3 as shown in Section~\ref{sec:cycles}). Using Theorem~\ref{thm:pathwidth}, we can get upper bounds better than 6 for the threshold dimension for a special kind of outerplanar graph.

The weak dual $G^*$ of an outerplanar graph $G$, given some planar embedding of $G$, is its dual graph with the vertex corresponding to the outer face removed. That is, $V(G^*)$ is the set of internal faces of $G$ and there is an edge between two internal faces $f$ and $f'$ in $G^*$ if and only if they share an edge in $G$.
Let $G$ be a 2-connected outerplanar graph whose weak dual is a path. Construct another outerplanar graph $G'$ by adding edges to $G$ such that every internal face of $G'$ is a triangle and the weak dual of $G'$ is also a path. It can be seen that $G'$ is an interval graph (see Theorem~8.1 in~\cite{golumbic}) with no clique of size more than 3 (as $G'$ is outerplanar). It follows that $\pw(G)\leq\pw(G')\leq 2$ (see Theorem~29 in~\cite{Bodlaender98}). Then we can use Theorem~\ref{thm:pathwidth} to get the following result, which is a generalization of Remark~\ref{rem:dimcographcycle}.

\begin{Corollary}\label{cor:weakdualpath}
If $G$ is a 2-connected outerplanar graph whose weak dual is a path, then $\dim_{COG}(G)\leq \dim_{TH}(G)\leq 3$.
\end{Corollary}

\section{Upper bounds for cograph dimension using vertex partitions}\label{sec:vertpart}

In this section, we borrow a technique from~\cite{kratochvil1994intersection} which can be used for graphs whose vertex set  can be partitioned in such a way that each pair of parts induces a subgraph with a bounded cograph dimension.
The following lemma that we use is a slight variation of a lemma that appears in~\cite{kratochvil1994intersection}. This technique and its role in connecting certain intersection dimensions of a graph with its acyclic and star chromatic numbers also appears in~\cite{aravindcrs}.

Recall that $G_1+G_2$ denotes the join of two graphs $G_1$ and $G_2$.

Let $\alpha:\mathbb{N}^+\rightarrow\mathbb{N}^+$ be the function  $\alpha(x)=\left\{\begin{array}{ll}x&\mbox{if }x\mbox{ is odd}\\x-1&\mbox{if }x\mbox{ is even}\end{array}\right.$.

\begin{Lemma}\label{lem:genpart}
Let $\mathcal{H}$ be a hereditary class of graphs which is closed under the join operation. If $G$ is a graph whose vertices can be partitioned into $k$ parts $V_1,V_2,\ldots,V_k$ in such a way that for any $i,j\in\{1,2,\ldots,k\}$, $\dim_{\mathcal{H}}(G[V_i\cup V_j])\leq t$, then $\dim_{\mathcal{H}}(G)\leq \alpha(k)t$.
\end{Lemma}
\begin{proof}
We shall begin by noting that as $\mathcal{H}$ is hereditary, for any $i\in\{1,2,\ldots,k\}$, $\dim_{\mathcal{H}} (G[V_i])\leq t$.
Consider the complete graph $K_k$ whose vertices are labelled from $c_1$ to $c_k$. For an edge $e=c_ic_j$ in this complete graph, we denote by $G_e$ the graph $G[V_i\cup V_j]$. For a matching $M$ in the complete graph, let $U_M$ denote the set of vertices of the complete graph that are not matched by $M$, i.e, the set of vertices of the complete graph that have no edge of $M$ incident on them. We also define $G_M=\Sigma_{e\in M} G_e + \Sigma_{c_i\in U_M} G[V_i]$. Since $\mathcal{H}$ is closed under the join operation, it follows by Observation~\ref{obs:join} that $\dim_{\mathcal{H}} (G_M)\leq t$. It is also easy to see that $G_M$ is a supergraph of $G$. Now consider a proper edge colouring of the complete graph $K_k$ using $\chi'(K_k)$ colours. This colouring can be seen as a partitioning of the edge set of the complete graph into matchings $M_1,M_2,\ldots,M_{\chi'(K_k)}$. Consider $uv\notin E(G)$. If there exists $i$ such that $u,v\in V_i$, then $uv\notin E(G_{M_j})$ for any $j$. If there is no such $i$, then surely there exists $i,j$ such that $u\in V_i$ and $v\in V_j$. Let $M_r$ be the matching that contains the edge $c_ic_j$ of the complete graph. Then, it can be seen that $uv\notin E(G_{M_r})$. This allows us to conclude that $G=\bigcap_{1\leq i\leq\chi'(K_k)} G_{M_i}$. This implies that $\dim_{\mathcal{H}}(G)\leq \chi'(K_k)t$. The proof is completed by noting the well known fact that $\chi'(K_k)=\alpha(k)$.
\end{proof}

The \emph{acyclic chromatic number} of a graph $G$, denoted by $\chi_a(G)$, is the minimum number of colours required to properly vertex colour $G$ such that the union of any two colour classes induces a forest in $G$.

We would like to note that some generalized variants of Theorems~\ref{thm:acyc} and~\ref{thm:star} appear in the work of Aravind and Subramanian~\cite{aravindcrs}. Theorem~9 and Corollary~11(b) of~\cite{aravindcrs} gives upper bounds on the intersection dimension of a graph $G$ with respect to a hereditary class that is closed under the disjoint union and join operations. We note here that any such class has to be the same as the class of cographs and therefore these bounds are in fact upper bounds on the cograph dimension of a graph $G$. The following theorem can also be obtained by combining Theorem~9 of~\cite{aravindcrs} and Corollary~\ref{cor:cographdimforests} in this paper.
\begin{Theorem}\label{thm:acyc}
For any graph $G$, $\dim_{COG}(G)\leq 2\cdot\alpha(\chi_a(G))$.
\end{Theorem}
\begin{proof}
Let $k=\chi_a(G)$.
Consider an acyclic vertex colouring of $G$ using $k$ colours. Let $V_1,V_2,\ldots,V_k$ denote the colour classes into which $V(G)$ gets partitioned by the colouring. Note that since for any $i,j\in\{1,2,\ldots,k\}$, $G[V_i\cup V_j]$ is a forest, we have from Corollary~\ref{cor:cographdimforests} that $\dim_{COG}(G[V_i\cup V_j])\leq 2$. Now from Lemma~\ref{lem:genpart}, we have the theorem.
\end{proof}

By Borodin's Theorem~\cite{borodin1979acyclic}, the acyclic chromatic number of any planar graph is at most 5. Therefore, we have the following corollary.

\begin{Corollary}\label{cor:cogdimplanar}
For any planar graph $G$, $\dim_{COG}(G)\leq 10$.
\end{Corollary}
\medskip

A \emph{star colouring} of a graph $G$ is a proper vertex colouring of $G$ such that the union of any two colour classes induces in $G$ a forest whose every component is a star (such a forest is called a ``star forest''). In other words, it is a proper vertex colouring of $G$ such that every path on 4 vertices in $G$ needs at least three colours. The minimum number of colours required in any star colouring of a graph $G$ is called its \emph{star chromatic number}, denoted by $\chi_s(G)$. It follows from Corollary~11(b) of~\cite{aravindcrs} that for any graph $G$, $\dim_{COG}(G)\leq\chi_s(G)$. The following theorem essentially states this, with the small improvement that $\chi_s(G)$ is replaced with $\alpha(\chi_s(G))$.

\begin{Theorem}\label{thm:star}
For any graph $G$, $\dim_{COG}(G)\leq\alpha(\chi_s(G))$.
\end{Theorem}
\begin{proof}
Let $k=\chi_s(G)$. Therefore, $V(G)$ can be partitioned into $k$ sets $V_1,V_2,\ldots,V_k$ such that for any $i,j\in\{1,2,\ldots,k\}$, $G[V_i\cup V_j]$ is a star forest. Since star forests are cographs, we have that for any $i,j\in\{1,2,\ldots,k\}$, $\dim_{COG}(G[V_i\cup V_j])=1$. Now from Lemma~\ref{lem:genpart}, we have the theorem.
\end{proof}

For planar graphs with lower bounds on girth, better bounds on the acyclic vertex colouring number and star chromatic number are known.
Every planar graph of girth at least 5 and 7 can be acyclically vertex coloured with 4 and 3 colours respectively~\cite{borodin1999acyclic}, implying that the cograph dimension of these graphs is at most 8 and 6 respectively. 
Further, every planar graph with girth at least 6, 7, 8, 9 and 13 can be star coloured with 8, 7, 6, 5 and 4 colours respectively~\cite{timmons2008star,kundgen2010star,bu2009star}. We therefore can use Theorems~\ref{thm:acyc} and~\ref{thm:star} to get the following.

\begin{Corollary}\label{cor:cographplgirth}
For any planar graph $G$ with girth $g$:
\begin{enumerate}
\renewcommand{\labelenumi}{(\roman{enumi})}
\item if $g\geq 5$, then $\dim_{COG}(G)\leq 8$
\item if $g\geq 6$, then $\dim_{COG}(G)\leq 7$,
\item if $g\geq 7$, then $\dim_{COG}(G)\leq 6$,
\item if $g\geq 8$, then $\dim_{COG}(G)\leq 5$, and
\item if $g\geq 13$, then $\dim_{COG}(G)\leq 3$.
\end{enumerate}
\end{Corollary}

As there exist planar bipartite graphs with star chromatic number at least 8~\cite{kierstead2009star}, we cannot hope to use this technique to get an improvement over the bound given in Corollary~\ref{cor:plbipthdim}.

A number of upper bounds are known for the acyclic chromatic number and star chromatic number of various special classes of graphs. These bounds can be used with Theorems~\ref{thm:acyc} and~\ref{thm:star} to establish upper bounds on the cograph dimension of graphs belonging to these classes.

\bigskip

\noindent\textbf{A note on the tightness of the upper bound of Fertin et al.}
The results of Fertin, Raspaud and Reed~\cite{fertin2004star} imply that for outerplanar graph (in fact, any partial 2-tree) $G$, $\chi_s(G)\leq 6$. Thus Theorem~\ref{thm:star} does not give us a bound better than the one given by Corollary~\ref{cor:cogdimp2tree}. An upper bound better than 6 for the star chromatic number of outerplanar graphs would have been helpful, but there are outerplanar graphs with star chromatic number equal to 6, and an example is given in~\cite{fertin2004star}. However, that example contains 48 vertices and 93 edges, and was shown to have star chromatic number 6 using a computer check. We give a simpler outerplanar graph which can be proven (without the need of a computer) to have star chromatic number at least 6.

 \begin{figure}[h]
 \begin{center}
  \renewcommand{\vertexset}{(x1,1,3),(x2,3,3),(x3,3,1),(x4,1,1),
  (y11,0,3),(y12,0.33,3.33),(y13,0.66,3.66),(y14,1,4),
  (y21,3,4),(y22,3.33,3.66),(y23,3.66,3.33),(y24,4,3),
  (y31,4,1),(y32,3.66,0.66),(y33,3.33,0.33),(y34,3,0),
  (y41,0,1),(y42,0.33,0.66),(y43,0.66,0.33),(y44,1,0)}
  \renewcommand{\edgeset}{(x1,x2),(x2,x3),(x3,x4),(x4,x1),(x2,x4),
  (x1,y11),(x1,y12),(x1,y13),(x1,y14),
  (x2,y21),(x2,y22),(x2,y23),(x2,y24),
  (x3,y31),(x3,y32),(x3,y33),(x3,y34),
  (x4,y41),(x4,y42),(x4,y43),(x4,y44),
  (y11,y12),(y12,y13),(y13,y14),
  (y21,y22),(y22,y23),(y23,y24),
  (y31,y32),(y32,y33),(y33,y34),
  (y41,y42),(y42,y43),(y43,y44)}
  \renewcommand{\defradius}{.1}
  \begin{tikzpicture}[scale=0.75]
  \drawgraph
  \node[below left] at (\xy{x1}) {$x$};
  \node[below right] at (\xy{x2}) {$y$};
  \node[above right] at (\xy{x3}) {$w$};
  \node[above left] at (\xy{x4}) {$z$};
  \node[above left=15] at (\xy{x1}) {$P_x$};
   \node[above right=14] at (\xy{x2}) {$P_y$};
   \node[below right=14] at (\xy{x3}) {$P_w$};
   \node[below left=14] at (\xy{x4}) {$P_z$};
   \path (\xy{y11}) -- (\xy{y14}) node [midway, above, sloped] {$\overbrace{\hspace{.5in}}$};
   \path (\xy{y21}) -- (\xy{y24}) node [midway, above, sloped] {$\overbrace{\hspace{.5in}}$};
   \path (\xy{y31}) -- (\xy{y34}) node [midway, below, sloped] {$\underbrace{\hspace{.5in}}$};
   \path (\xy{y41}) -- (\xy{y44}) node [midway, below, sloped] {$\underbrace{\hspace{.5in}}$};
  \end{tikzpicture}
 \end{center}
  \caption{The graph $G$.}
  \label{fig:tightexample}
 \end{figure}
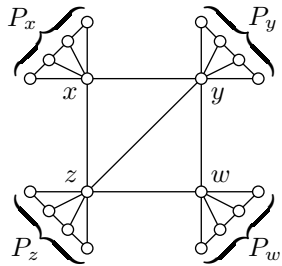

\begin{Theorem}\label{thm:startight}
The outerplanar graph $G$ shown in Figure~\ref{fig:tightexample} has star chromatic number at least 6.
\end{Theorem}
\begin{proof}
Suppose for the sake of contradiction that $G$ has a star colouring $c$ using the colours 1, 2, 3, 4 and 5. Since every star coloring is also a proper vertex coloring, we can assume without loss of generality that $c(x)=1$, $c(y)=2$ and $c(z)=3$, and further that no vertex in $P_x$ has color 1. As $c$ is a star coloring, the vertices in $P_x$ receive at least 3 different colors, implying that at least one of the colors 2 or 3 is present in the path $P_x$. Let us assume without loss of generality that one of the vertices in $P_x$, which we shall denote by $x_2$, is colored 2. We now get that $c(w)\neq 1$, as otherwise, $x_2xyw$ would be a bicolored $P_4$. So we can assume without loss of generality that $c(w)=4$. If a vertex $y_1$ in $P_y$ is colored 1, then we would have the bicolored $P_4$ $x_2xyy_1$. Since $P_y$ also cannot be bicolored, we get that there is a vertex $y_i$ in $P_y$ that is colored $i$, for each $i\in\{3,4,5\}$. Now if there is a vertex $w_2$ in $P_w$ that is colored 2, then we will have a bicolored $P_4$ $y_4yww_2$, which is a contradiction. This implies that there is a vertex $w_3$ in $P_w$ that is colored 3. Similarly, since if there is a vertex $z_2$ in $P_z$ that is colored 2, there will be the bicolored $P_4$ $y_3yzz_2$, we get that there is a vertex $z_4$ in $P_z$ that is colored 4. Now the path $w_3wzz_4$ is a bicolored $P_4$, contradicting the fact that $c$ is a star coloring.
\end{proof}

\section{Conclusion}\label{sec:conclusion}
Table~\ref{tab:results} lists the upper bounds on the cograph dimension and threshold dimension of some of the subclasses of planar graphs that were studied.
\begin{table}[h]
\centering
 \begin{tabular}{|l|c|c|}
 \hline
 \multicolumn{1}{|c|}{Graph $G$ is} & $\dim_{COG}(G)\leq$ & $\dim_{TH}(G)\leq$ \\
 \hline
 Planar & 10 {\tiny (Corollary~\ref{cor:cogdimplanar})} & 12 {\tiny (Corollary~\ref{cor:thresholdplanar})}\\
 \hline
 \hspace{0.5in} girth $\geq 4$ & 9 {\tiny (Corollary~\ref{cor:thresholdtrianglefreeplanar})}&9 {\tiny (Corollary~\ref{cor:thresholdtrianglefreeplanar})}\\
 \hline
 \hspace{0.5in} $8\geq$ (girth $g$) $\geq 5$ & $13-g$ {\tiny (Corollary~\ref{cor:cographplgirth})}&9 {\tiny (Corollary~\ref{cor:thresholdtrianglefreeplanar})}\\
 \hline
 \hspace{0.5in} girth $\geq 13$ & 3 {\tiny (Corollary~\ref{cor:cographplgirth})}&9 {\tiny (Corollary~\ref{cor:thresholdtrianglefreeplanar})}\\
 \hline
 Partial 2-tree & 4 {\tiny (Corollary~\ref{cor:cogdimp2tree})}& 9 {\tiny (Corollary~\ref{cor:thresholdp2tree})}\\
 \hline
 Outerplanar & 4 {\tiny (Corollary~\ref{cor:cogdimp2tree})}& 6 {\tiny (Corollary~\ref{cor:thresholdouterplanar})}\\
 \hline
 \hspace{0.5in} weak dual is a path&3 {\tiny (Corollary~\ref{cor:weakdualpath})}&3 {\tiny (Corollary~\ref{cor:weakdualpath})}\\
 \hline
 Planar bipartite & 4 {\tiny (Corollary~\ref{cor:plbipthdim})}& 4 {\tiny (Corollary~\ref{cor:plbipthdim})}\\
  \hline
 Forest & 2 {\tiny (Corollary~\ref{cor:cographdimforests})}& 4 {\tiny (Corollary~\ref{cor:plbipthdim})}\\
 \hline
 \end{tabular}
 \caption{Upper bounds on the cograph and threshold dimensions of some subclasses of planar graphs}
 \label{tab:results}
 \end{table}

\noindent We conclude with some open questions.
\medskip

\noindent\textbf{Question.} Can the upper bounds shown for the cograph dimension and threshold dimension for any of the classes of graphs studied be improved?
\medskip

\noindent\textbf{Question.} Does there exist a linear upper bound on the threshold dimension in terms of the treewidth of the graph?
\medskip

\noindent\textbf{Question.} Does there exist a planar graph whose cograph dimension is more than 3?
\bigskip

\noindent\textbf{Acknowledgements.}
The authors would like to thank T.~Karthick for his help in simplifying the proof of Theorem~\ref{thm:startight}.

\bibliography{cogdim}

\end{document}